\documentclass[conference]{IEEEtran}
\IEEEoverridecommandlockouts

\usepackage[cmex10]{amsmath}
\usepackage{amsthm}
\usepackage{amssymb}
\usepackage{pifont}
\usepackage{graphicx}
\usepackage{subfigure}

\newtheorem{lemma}{Lemma}
\newtheorem{Theorem}{Theorem}
\newtheorem{Remark}{Remark}

\begin{document}

\title{Downlink Goodput Analysis for D2D Underlaying Massive MIMO Networks}

\author{\IEEEauthorblockN {Zezhong Zhang\IEEEauthorrefmark{1}
		Zehua Zhou\IEEEauthorrefmark{1},
		Rui Wang\IEEEauthorrefmark{1}, and
		Yang Li\IEEEauthorrefmark{2}, \thanks{This work is supported by NSFC 61401192 and Shenzhen Science and Technology Innovation Committee JCYJ20160331115457945.} }
	\IEEEauthorblockA{
		\IEEEauthorrefmark{1}Department of Electrical and Electronic Engineering, The Southern University of Science and Technology, China\\
		\IEEEauthorrefmark{2}Department of Information and Communication Engineering, Xi'an Jiaotong University, China\\
		Email:  \{zhangzz, zhouzh\}@mail.sustc.edu.cn, wang.r@sustc.edu.cn, liyang.ei@stu.xjtu.edu.cn}}
\maketitle

\maketitle

\begin{abstract}
The performance of downlink massive multiple-input-multiple-output networks with co-channel device-to-device communications is investigated in this paper. Specifically, we consider a cellular network with sufficient number of antennas at the base station and typical hexagonal cell coverage, where the cell users and device-to-device transmitters are randomly and uniformly distributed. To obtain the analytical expressions of system-level performance, the asymptotic signal-to-interference ratios for both downlink and device-to-device links are first obtained, which depend on the pathloss and small-scale fading of the interference channels. Since these information may not be available at the service base station or device-to-device transmitters, there exists a chance of packet outage. Therefore, we continue to derive the closed-form approximation of the average goodput, which measures the average number of information bits successfully delivered to the receiver. Hence, the system design trade-off between downlink and co-channel device-to-device communications can be investigated analytically. Moreover, the performance region in which the co-channel device-to-device communications could lead to better overall spectral efficiency can be obtained. Finally, it is shown by simulations that the analytical results matches the actual performance very well.

\end{abstract}

\section{Introduction}

Massive multiple-input multiple-output (MIMO) is an efficient technique to boost the spectral efficiency. However, as elaborated in the existing literature, the issue of pilot contamination \cite{PilotContamination}, which refers to the undiminished inter-cell interference caused by pilot reuse, may severely degrade its performance. Although there have been significant research efforts devoted to address the pilot contamination issue \cite{InterferenceTDD,2016ICCS} and some techniques proposed to mitigate the pilot contamination~\cite{2014Globecom,2015Globecom}, most of the works handle the inter-cell interference from the conventional cellular network point of view, instead of novel network topologies.

Recently, the deployment of wireless cache nodes in cellular networks has drawn significant research attentions \cite{WirelessNode1, WirelessNode2}, where the device-to-device (D2D) links has to be introduced into cellular network. The co-channel deployment of D2D communications and cellular networks was studied in some literature \cite{D2DCellular1, D2DCellular2}, where it is shown that the D2D transmission reusing the cellular spectrum may cause severe interference. To alleviate the interference, D2D underlay massive MIMO cellular networks has been proposed by exploiting spatial degrees of freedom at the base station (BS). For example, in \cite{NovelMMSE}, the authors proposed a pilot reuse strategy for D2D receivers and a novel interference-aided minimum mean square error (MMSE) detector to suppress the D2D-to-cellular interference. In \cite{PilotAllocation}, a novel revised graph coloring-based pilot allocation (RGCPA) algorithm was proposed for pilot allocation, and an iterative scheme was adopted to minimize the transmission power of D2D links. In order to evaluate the overall effect of D2D links on cellular network, a system-level performance analysis is necessary. The interplay between massive MIMO uplink transmission and co-channel D2D transmission has been studied in \cite{Interplay}. However, the conclusion of its analysis cannot be directly applied in downlink. Moreover, the authors assume that the antenna number at the BS is infinity, which is not practical\footnote{In massive MIMO systems, the typical number of antennas at the BS is up to a few hundred, which is large but not infinity.}. As a result, performance analysis of massive MIMO downlink transmission with co-channel D2D links is still open.

In this paper, we would like to shed some light on the above issue by analysing the performance trade-off between massive MIMO downlink and co-channel D2D transmissions. Specifically, we consider a massive MIMO network with typical hexagonal cell structure and random distribution users and D2D links. The D2D links may refer to the direct transmission from wireless cache nodes to cell users, which off-load the traffic from service BSs. In the analysis, we first obtain the asymptotic expressions of downlink and D2D signal-to-interference ratios (SIRs) for sufficiently large (but finite) number of antennas at the BS. These expressions depend on the pathloss and small-scale fading of interference channels, which may be unknown to the service BS or D2D transmitters. For example, the interference to downlink users comes from neighbouring BSs and nearby D2D transmitters, and the channel condition of those interference is not easy to obtain at the service BS. When the randomness of user locations and small-scale fading are considered, it is possible that the transmission data rate is greater than the channel capacity, leading to the {\em packet outage}. Hence We use the average goodput \cite{Goodput}, which measures the average number of bits successfully delivered to the receiver, as performance metric. We derive the approximated expressions of the downlink and D2D goodput, based on which the performance trade-off between massive MIMO downlink and D2D can be evaluated. It is shown by numerical simulations that the analytical results matches the actual performance very well.


\section{System Model}
\subsection{Network Model}

We consider a typical cellular network with $C$ hexagonal cells where the radius of each cell is $R$, as illustrated in Figure. \ref{Network}. Each cell consists of a base station equipped with $M$ antennas, $K$ active single-antenna downlink users and $D$ active D2D transmitters. BSs and D2D transmitters transmit with constant powers $P_{b}$ and $P_{d}$, respectively. The D2D links may refer to the data delivery between the downlink users and wireless cache nodes. For example, the desired data of downlink users is found in one cache node nearby, and then a direct D2D communication link is established. Since the focus of this paper is  on the physical-layer SIR and throughput analysis, the establishment of D2D links is outside the scope of this work. Since massive MIMO technology is considered, $ M $ is sufficiently large, e.g., a few hundred. The downlink users and D2D transmitters are uniformly and independently distributed. Without loss of generality, we investigate the performance of the first cell while other cells are all interfering cells. The $ j $-th downlink user of the $ i $-th cell is referred to as the $ (i,j) $-th downlink user.

The massive MIMO network is working in time-division duplex (TDD) mode. Thus it is assumed that the downlink channel of downlink users is estimated from their uplink pilot transmission within the same coherent fading block. Moreover, in order to improve the overall spectrum efficiency, we consider the co-channel deployment of D2D and downlink transmission, i.e., the D2D transmitters use the same spectrum as the cellular network. Note that the coexistence issue of massive MIMO uplink and D2D communications has been investigated in \cite{Interplay}, the focus of this paper is put on sharing the downlink transmission opportunities with D2D links. All D2D transmitters and receivers are equipped with single antenna. The $ k $-th D2D link (transmitter or receiver) of the $ i $-th cell is referred to as the $ (i,k) $-th D2D link (transmitter or receiver). The notations of downlink and D2D transmissions are summarized below.
\begin{itemize}
	\item ${\bf{h}}_{l,k}^i, {\bf{v}}_{l,j}^i \in {\mathcal{C}^{1 \times M}}$ represent the downlink channel vectors from the $i$-th BS to the $k$-th downlink user and $j$-th D2D receiver in the $l$-th cell, respectively. Each component of ${\bf{h}}_{l,k}^i$ and ${\bf{v}}_{l,j}^i$ is complex Gaussian with mean zero and variance $\rho_{l,k}^i$ and $\rho_{l,j}^i$ respectively.  $\rho _{l,k}^i = {\left( {s_{l,k}^i} \right)^{ - \sigma }}$ and $\rho _{l,j}^i = {\left( {s_{l,j}^i} \right)^{ - \sigma }}$ are the pathloss from the $i$-th BS to the $(l,k)$-th downlink user and $(l,j)$-th D2D receiver, where ${s_{l,k}^i}$ and ${s_{l,j}^i}$ are the distances from the $i$-th BS to the $(l,k)$-th downlink user and $(l,j)$-th D2D receiver. $\sigma$ is the pathloss exponent between BSs and users. ${{\bf{ H}}_i^l\in {\mathcal{C}^{K \times M}}}$ and ${{\bf{V}}_i^l\in {\mathcal{C}^{K \times M}}}$ are the aggregation of ${\bf{h}}_{l,j}^i$ and ${\bf{v}}_{l,j}^i$ within one cell.
	\item ${{g}}_{l,j}^{i,m}, {{u}}_{l,k}^{i,m}$ represents the downlink channel vector from the $(i,m)$-th D2D transmitters to the $j$-th D2D receiver and $k$-th downlink user in the $l$-th cell, respectively. ${{g}}_{l,j}^{i,m}$ and ${{u}}_{l,k}^{i,m}$ are complex Gaussian with mean zero and variance $\rho_{l,j}^{i,m}$ and $\rho_{l,k}^{i,m}$, respectively.  $\rho _{l,j}^{i,m} = {\left( {d_{l,j}^{i,m}} \right)^{ - \kappa }} $ and $\rho _{l,k}^{i,m} = {\left( {d_{l,k}^{i,m}} \right)^{ - \kappa }} $ are the pathloss from the $(i,m)$-th D2D transmitter to the $(l,j)$-th D2D receiver and $(l,k)$-th downlink user, where ${d_{l,j}^{i,m}}$ and ${d_{l,k}^{i,m}}$ are the distances from the $(i,m)$-th D2D transmitter to the $(l,j)$-th D2D receiver and $(l,k)$-th downlink user. $\kappa$ is the pathloss exponent between users. ${{\bf{G}}_i^l\in {\mathcal{C}^{K \times D}}}$ and ${{\bf{U}}_i^l\in {\mathcal{C}^{K \times D}}}$ are aggregation of ${{g}}_{l,j}^{i,m}$ and ${{u}}_{l,j}^{i,m}$.
	\item  ${{\bf{x}}_{l,k}^p}$ is the pilot sequence of the $(l,k)$-th downlink user. We have $\left| {{\bf{x}}_{l,k}^p{{({\bf{x}}_{l,m}^p)}^H}} \right| = 0,\forall k \ne m$ and $\left| {{\bf{x}}_{l,k}^p{{({\bf{x}}_{i,j}^p)}^H}/{L_p}} \right| = \frac{{{P_{u}}}}{{\sqrt {{L_p}} }},\forall i \ne l,$ where $P_{u}$ is the transmission power of each mobile user and $L_p$ represents the pilot length in the uplink.  ${{\bf{X}}_l^p}$ is the aggregation of pilot sequences from active downlink users in the $l$-th cell. .
\end{itemize}

\begin{figure}
	\centering
	\includegraphics[height=4.5cm,width=5cm]{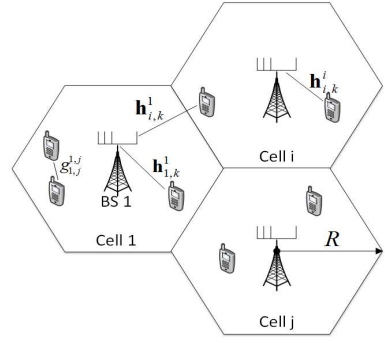}
	\caption[simulation]{Illustration of hexagonal cellular network with radius R, where Cell 1 is the target cell.} \label{Network}
\end{figure}

\begin{Remark}The coexistence SIR analysis of massive MIMO uplink transmission and D2D transmission in \cite{Interplay} cannot be applied in downlink scenario, as the source of interference is completely different. In this paper, we also propose a new analytical framework to evaluate the asymptotic goodput performance with sufficiently large (but finite) number of antennas at BS $M$. Note that the approach introduced in \cite{Interplay} is for infinite M, which may not be accurate when M is only a few hundred. Moreover, we use Gaussian approximation to obtain a simple closed-form expression of the cumulative distribution function (CDF) of SIR, based on which we also derive the average goodput as the performance metric so that the potential packet outage can be counted. These results cannot be obtained with the approach in \cite{Interplay}.
\end{Remark}

\subsection{Channel Model}

Since D2D links share the downlink transmission opportunities, they are silent in the uplink subframe. Thus in the channel estimation phase (as illustrated in Fig. \ref{frame}), the received signal of the $i$-th BS is given by
\begin{equation}
{\bf{Y}}_i^p = {{{\left( {{\bf{H}}_i^i} \right)}^H}{\bf{X}}_i^p} + \sum\limits_{\forall l \neq i} {{{\left( {{\bf{H}}_l^i} \right)}^H}{\bf{X}}_l^p}. \nonumber
\end{equation}

\begin{figure}
	\centering
	\includegraphics[height=2.5cm,width=6cm]{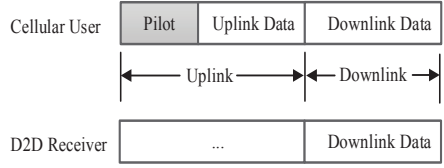}
	\caption[frame]{Illustration of channel model for both downlink and D2D transmission.} \label{frame}
\end{figure}

 With match filter, the estimated uplink channel can be written as
\begin{eqnarray}
{\left( {{\bf{\hat H}}_i^i} \right)^H} &=& {\bf{Y}}_i^p{\left( {{\bf{X}}_i^p} \right)^H}/{L_p}{P_u}\nonumber\\
&=& {\left( {{\bf{H}}_i^i} \right)^H} + \underbrace {\sum\limits_{\forall l \ne i} {{{\left( {{\bf{H}}_l^i} \right)}^H}/\sqrt {{L_p}} } }_{{{\left(\Delta {{\bf{H}}_i^i} \right)}^H}},
\end{eqnarray}
where ${{\bf{\Delta H}}_i^i}$ is the channel estimate error of the $i$-th BS. The first term is the desired CSI and the second term is the inter-cell interference due to pilot reuse (pilot contamination).

In the downlink transmission, the aggregated received signal of the users in the $i$-th cell (as illustrated in Fig. \ref{Network}) is
\begin{eqnarray}
		{{\bf{Y}}_i}\!\!\! &=&\!\! {\bf{H}}_{i}^i{\bf{X}}_i^d + \sum\limits_{l \ne i} {{\bf{H}}_{_{i}}^l {\bf{X}}_l^d}+\sum\limits_{\forall l} {{\bf{U}}_i^l{\bf{S}}_l},
\end{eqnarray}
where ${\bf{U}}_i^l{\bf{S}}_l$ is due to the interference from D2D transmitters. ${{\bf{x}}_{l,k}^d}$ and ${{\bf{s}}_{l,k}^d}$ are the downlink signals for the $k$-th downlink user and the $k$-th D2D receiver in the $l$-th cell. ${{\bf{X}}_l^d}$ and ${{\bf{S}}_l}$ are the aggregations of downlink data blocks for downlink and D2D receivers in the $l$-th cell with average power $P_b$ and $P_d$.

Moreover, the aggregated received signal of the D2D receivers in the $i$-th cell during the downlink subframe can be written as
\begin{eqnarray}
{{\bf{Y}}_i^{D2D}} = \sum\limits_{\forall l}{{\bf{G}}_i^l{\bf{S}}_l + {\bf{V}}_i^l{\bf{X}}_l^d}.
\end{eqnarray}
We neglect the effect of noise because the noise is much smaller than the interference.

Since the channel is estimated at the BS before downlink transmission, the zero-forcing precoder based on imperfect channel state information (CSI) is designed as $${{\bf{P}}_i} = {\left( {{\bf{\hat H}}_i^i} \right)^H}{\left( {{\bf{\hat H}}_i^i{{\left( {{\bf{\hat H}}_i^i} \right)}^H}} \right)^{ - 1}}$$ Denoting the aggregated downlink data block before precoding as ${\bf{W}}_i^d$, the downlink signal for downlink users in the $i$-th cell is
\begin{eqnarray}
{{\bf{Y}}_i}\!\!\!\! &=&\!\!\! {\bf{H}}_{i}^i{{\bf{P}}_i}{\bf{W}}_i^d + \sum\limits_{l \ne i} {{\bf{H}}_{_{i}}^l{{\bf{P}}_l}{\bf{W}}_l^d}+\sum\limits_{\forall l}  {{\bf{U}}_i^l{\bf{S}}_l} \nonumber \\
&=&\!\!\!\!\!\! \sum\limits_{\forall l}\! {{\bf{H}}_i^l{{\left( {{\bf{\hat H}}_l^l} \right)}^H}\!\!{{\left(\! {{\bf{\hat H}}_l^l{{\left( {{\bf{\hat H}}_l^l} \right)\!}^H}} \right)\!}^{ - 1}}\!\!\!\!{\bf{W}}_l^d} \! +\! \sum\limits_{\forall l}\! {{\bf{U}}_i^l{\bf{S}}_l}  \label{MatrixCellular}
\end{eqnarray}
and the SIR of the $(i,k)$-th downlink user can be given by
\begin{eqnarray}
{\gamma _{i,k}} = \frac{1}{{{{\left\|\Delta  {{\bf{h}}_{i,k}^i{{\bf{P}}_i}} \right\|}^2}\! +\! \sum\limits_{l \ne i} {{{\left\| {{\bf{h}}_{i,k}^l{{\bf{P}}_l}} \right\|}^2}} \! +\! \sum\limits_{\forall l,m} {{{\left\| {u_{i,k}^{l,m}} \right\|}^2}}\frac{{{P_d}}}{{{P_b}}} }},\label{CellularExpression}
\end{eqnarray}
where $\Delta  {\bf{h}}_{i,k}^i$ is the channel estimate error of ${\bf{h}}_{i,k}^i$.
Similarly, the downlink signal for D2D receivers in the $i$-th cell can be rewritten as
\begin{eqnarray}
{{\bf{Y}}_i^{D2D}} = \sum\limits_{\forall l} { {\bf{G}}_i^l{\bf{S}}_l^d +  {\bf{V}}_i^l{{\bf{P}}_l}{\bf{W}}_l^d}.  \label{MatrixD2D}
\end{eqnarray}
It is assumed that the D2D transmitter has no knowledge of CSI, but the receiver has perfect CSI. Hence the SIR of the $(i,k)$-th D2D receiver can be given by
\begin{eqnarray}
{\gamma _{i,k}^{D2D}} = \frac{{{{\left\| {g_{i,k}^{i,k}} \right\|}^2}}}{{\sum\limits_l {{{\left\| {{\bf{v}}_{i,k}^l{{\bf{P}}_l}} \right\|}^2}}\frac{{{P_b}}}{{{P_d}}}  + \sum\limits_{(l,m) \ne (i,k)} {{{\left\| {g_{i,k}^{l,m}} \right\|}^2}} }}.\label{D2DExpression}
\end{eqnarray}
Note that the D2D transmitter does not know $\gamma_{i,k}^{D2D}$, it has to determine the data rate according to the statistics of $\gamma_{i,k}^{D2D}$. Thus goodput is introduced in the following section as the performance metric.

\subsection{Average Goodput}
It can be observed from (\ref{CellularExpression}) that the downlink SIR is determined by channel estimation error $\Delta {\bf{h}}_{i,k}^i$ and CSI of interference channel ${\bf{h}}_{i,k}^l$, which may be unknown to the service BS. Hence the link capacity becomes random to the service BS, and it is possible that the scheduled downlink data rate may be larger than the channel capacity, which leads to the packet outage. In order to take the potential packet loss into consideration, we use the average goodput as the performance metric \cite{Goodput}. Given a scheduled downlink data rate $r_{i,k}$ for the $(i,k)$-th downlink user, the goodput is defined as
\begin{equation}
{\emph{g}_{i,k}} = {r_{i,k}}I\left( {{r_{i,k}} \le {{\log }_2}(1 + {\gamma _{i,k}})} \right),
\end{equation}
where $I\left( . \right)$ is is an indicator function with value $1$ if the event is true and $0$ otherwise. The the average goodput spanning all
possible channel realization is given by
\begin{equation}
{{\bar {\emph{g}}}_{i,k}} =
\mathbb{E}_{\bf{H}}\left[ {{\emph{g}_{i,k}}} \right] = {r_{i,k}}\Pr \left( {{r_{i,k}} \le {{\log }_2}(1 + {\gamma _{i,k}})} \right).
\end{equation}
In this paper, we consider the downlink transmission with a target
outage probability $\varepsilon $, thus
\begin{equation}
\Pr \left( {{r_{i,k}} \le {{\log }_2}(1 + {\gamma _{i,k}})} \right) = 1 - \varepsilon. \nonumber
\end{equation}
Define the SIR Threshold determined by $\varepsilon $ as $T_{i,k}$, where
\begin{equation}
\Pr \left( {{\gamma _{i,k}} \le {T_{i,k}}} \right) = \varepsilon .
\end{equation}
 The average goodput of the $(i,k)$-th downlink user becomes
\begin{equation}
{{\bar {\emph{g}}}_{i,k}} = (1 - \varepsilon ) \times {L} \times {\log _2}\left( {1 + {T_{i,k}}} \right),
\end{equation}
where $L$ is the number of total downlink symbols within one subband of a frame.

Similarly, for the $(i,j)$-th D2D receiver, given the outage probability $\varepsilon$, the average goodput becomes
\begin{equation}
{{\bar {\emph{g}}}_{i,k}^{D2D}} = (1 - \varepsilon ) \times {L} \times {\log _2}\left( {1 + {T_{i,k}^{D2D}}} \right),
\end{equation}
where ${T_{i,k}^{D2D}}$ satisfies
\begin{equation}
\Pr \left( {{\gamma _{i,k}^{D2D}} \le {T_{i,k}^{D2D}}} \right) = \varepsilon.
\end{equation}

\section{Goodput Performance analysis}
In this section, we first derive the asymptotic SIR expression for both downlink users and D2D receivers, and then provide the approximated expressions of goodput. First of all, we have the following lemma on asymptotic SIR.
\begin{lemma}[Aysmptotic Downlink SIR]
When the number of BS antennas is sufficiently large, the asymptotic SIR of the $(1,k)$-th downlink user is given by
\begin{eqnarray}
\!\!\!\!\!{\gamma _{1,k}} \!= \!\frac{{1}}{{\!\!\sum\limits_{\forall i \ne 1,j} \!{\left[\!\! {{{{\left( {\frac{{{\bf{h}}_{1,k}^i({\bf{h}}_{i,j}^i)^H}}{{M\rho _{i,j}^i}}} \!\right)\!}^2}}\!\!\! +\!\! \frac{1}{{{L_p}}}{{\left( {\frac{{\rho _{1,k}^i}}{{\rho _{i,j}^i}}} \right)\!}^2}}\! \right]\!\!}  +\!\! \sum\limits_{\forall i,j} {u_{1,k}^{i,j}{{(u_{1,k}^{i,j})}^H}\!}\frac{{{P_d}}}{{{P_b}}} }} \label{CellularSIR}
\end{eqnarray}
and the asymptotic SIR of the $(1,k)$-th D2D receiver is
\begin{eqnarray}
{\gamma _{1,k}^{D2D}} = \frac{{g_{1,k}^{1,k}{{(g_{1,k}^{1,k})}^H}}}{{\sum\limits_{\forall i,j} {{{\left( {\frac{{{\bf{v}}_{1,k}^i({\bf{h}}_{i,j}^i)^H}}{{M\rho _{i,j}^i}}} \right)}^2}}\frac{{{P_b}}}{{{P_d}}}  +\!\! \sum\limits_{\forall (i,j) \ne (1,k)}\!\! {g_{1,k}^{i,j}{{(g_{1,k}^{i,j})}^H}} }}.  \label{D2DSIR}
\end{eqnarray}
\end{lemma}
\begin{proof}
	Please refer to the Appendix A.
\end{proof}

From equations (\ref{CellularSIR}) and (\ref{D2DSIR}), we can observe that the downlink performance is related to downlink users' and D2D transmitters' locations as well as the channel variation. Some of them are unknown to the service BS or D2D transmitters. In this case, the service BS and D2D transmitter cannot ensure that the data rate is below the channel capacity. However, if we know the distribution of the SIR, the outage probability can be controlled by setting the data rate appropriately. Therefore, we derive the CDF of downlink SIR for both downlink users and D2D receivers in the following.

\begin{lemma}[CDF of SIRs]
	Given that the number of BS antennas M and the number of interfering downlink users $(C-1)K$ and D2D transmitters $CD$ are sufficiently large, the CDF of the $(1,k)$-th downlink user's downlink SIR can be approximated as
	\begin{equation}
\Pr \left[ {{\gamma _{1,k}} \le {\tau }} \right] \approx Q\left[ {\frac{{\frac{{1}}{{{\tau }}} - \sum\limits_{i \ne 1} {{\mu _{i,B}}{K}}  - \sum\limits_{\forall i} {{\mu _{i,D}}{D}} }}{{\sqrt {\sum\limits_{i \ne 1} {\sigma _{i,B}^2{K}}  + \sum\limits_{\forall i} {\sigma _{i,D}^2{D}} } }}} \right],
	\end{equation}
	where the $Q$-function is the tail probability of stantard normal distribution. $\mu_{i,B}$ and $\mu_{i,D}$ are denoted as the expectation of ${{{\left( {\frac{{{\bf{h}}_{1,k}^i({\bf{h}}_{i,j}^i)^H}}{{M\rho _{i,j}^i}}} \right)}^2} + \frac{1}{{{L_p}}}{{\left( {\frac{{\rho _{1,k}^i}}{{\rho _{i,j}^i}}} \right)}^2}}$ and ${u_{1,k}^{i,j}{{(u_{1,k}^{i,j})}^H}\frac{{{P_d}}}{{{P_b}}}}$, while $\sigma_{i,B}^2$ and $\sigma_{i,D}^2$ are their variances.
	
	The CDF of the $(1,k)$-th D2D receiver's downlink SIR can be approximated as
	\begin{eqnarray}\label{AsySIRd2d}
	\!\!\!&&\!\!\!\!\!\!\!\!\!\!\Pr \left[ {{\gamma _{1,k}} \le {\tau }} \right] \nonumber \\
	\!\!\!\approx&&\!\!\!\!\!\!\!\!\!\!\!\!  F\!{\left[\! {\tau \! \left(\! {\sum\limits_{i \ne 1}  \left( {{\mu _{i,B}}^\prime K\! +\! {\mu _{i,D}}D} \right) \!+\! {\mu _{1,D}}(D - 1)}\! \right)}\!\! \right]_{\lambda  = \rho _{1,k}^{1,k}}},
	\end{eqnarray}
where $F$-function represents the tail probability of exponential distribution with expectation ${\lambda  = \rho _{1,k}^{1,k}}$ and ${\mu _{i,B}}^\prime$ is the expectation of ${{{\left( {\frac{{ {\bf{v}}_{1,k}^i{\bf{h}}{{_{i,j}^i}^H}}}{{M\rho _{i,j}^i}}} \right)}^2}}\!\frac{{{P_b}}}{{{P_d}}}$.	
\end{lemma}
\begin{proof}
	Please refer to the Appendix B.
\end{proof}
\begin{Remark}
	Based on the uniform distribution of cellular users and D2D transmitters, $\mu_{i,B}$ and $\mu_{i,D}$ can be calculated as follows 	
	\begin{eqnarray}
	{\mu _{i,B}}\!\!\!\!\!\! &=&\!\!\!\!\!\! \mathop {\lim }\limits_{\zeta  \to \infty }\!\! \int_{\frac{{\sqrt 3 }}{2}R}^R {\frac{{12x{{\cos }^{ - 1}}\left( {\frac{{\sqrt 3 R}}{{2x}}} \right)}}{{ \left(\pi {{\zeta ^2}\!\! - \frac{{3\sqrt 3 }}{2}} \right){R^2}}}\!\!\!\left[\!\! {\frac{{\rho _{1,k}^i}}{{M\rho _{i,j}^i}} \!\!+\!\! \frac{1}{{{L_p}}}{{\left( {\frac{{\rho _{1,k}^i}}{{\rho _{i,j}^i}}} \right)\!\!}^2}} \right]\!\!dx}\nonumber \\
	&&\!\!\!\!\!\!+\!\! \mathop {\lim }\limits_{\zeta  \to \infty }\!\!\! \int_R^{\zeta R}\!\! {\frac{{2\pi x}}{{\left( \pi{{\zeta ^2}\! -\! \frac{{3\sqrt 3 }}{2}} \right){R^2}}}\!\!\left[\!\! {\frac{{\rho _{1,k}^i}}{{M\rho _{i,j}^i}}\!\! +\!\! \frac{1}{{{L_p}}}{{\left( {\frac{{\rho _{1,k}^i}}{{\rho _{i,j}^i}}} \right)\!\!}^2}} \right]\!\!dx}, \nonumber
	\end{eqnarray}
	\begin{eqnarray}
	\sigma _{i,B}^2 \!\!\!\!\!\!&=&\!\!\!\!\!\! \mathop {\lim }\limits_{\zeta \! \to \infty }\!\! \int_{\frac{\!{\sqrt 3 }}{2}R}^R\!\! {\frac{{12x{{\cos }\!^{ - 1}}\!\!\left(\!\! {\frac{{\sqrt 3 R}}{{2x}}}\! \right)}}{{\left(\! {\pi {\zeta ^2}\! -\! \frac{{3\sqrt 3 }}{2}}\! \right)\!\!{R^2}}}\!(\!\frac{{\rho _{1,k}^i}^2}{{M\!{\rho _{i,j}^i}}\!^2}\!\! + \!\!\frac{{\rho {{_{1,k}^i}^4}}}{{{L_p}\rho {{_{i,j}^i}\!^4}}}\!\! +\!\! \frac{{2\rho {{_{1,k}^i}^3}}}{{M\!{L_p}\rho {{_{i,j}^i}\!^3}}}\!)dx} \nonumber \\
	&&\!\!\!\!\!\!\!+\!\! \mathop {\lim }\limits_{\zeta \! \to \infty }\!\! \int_R^{\zeta R}\!\!\!\!\!  \frac{{2\pi x}}{{\left(\! {\pi {\zeta ^2}\!\! - \!\!\frac{{3\sqrt 3 }}{2}} \!\right)\!\!{R^2}}}(\!\frac{{\rho {{_{1,k}^i}^2}}}{{M\rho {{_{i,j}^i}\!^2}}}\!\! +\!\! \frac{{\rho {{_{1,k}^i}^4}}}{{{L_p}\rho {{_{i,j}^i}\!^4}}}\!\! +\!\! \frac{{2\rho {{_{1,k}^i}^3}}}{{M{L_p}\rho {{_{i,j}^i}\!^3}}}\!)dx. \nonumber
	\end{eqnarray}
${\mu _{i,D}}$ and $\sigma _{i,D}^2$ can be obtained Similarly. Note that the distance between D2D transmitters and receivers is much smaller than the cell radius, the distribution of D2D users can also be considered uniform, thus ${\mu _{i,B}^D}$, ${\mu _{i,D}^D}$ and ${\sigma {{_{i,B}^D}^2}}$, ${\sigma {{_{i,D}^D}^2}}$ can be calculated.
\end{Remark}

As a result, we have the following theorem on average goodput of downlink users and D2D receivers.

\begin{Theorem}[Average Goodput]
	With outage probability $\varepsilon $, the average goodput of the $(1,k)$-th downlink user is given as
	\begin{eqnarray}
	{\bar {\emph{g}}_{1,k}} = (1-\varepsilon) \times {L} \times {\log _2}\left( {1 + \frac{{1}}{{{Q^{ - 1}}\left( \varepsilon  \right)\sigma _1^{C} + \mu _1^{{C}}}}} \right), \label{goodputUser}
	\end{eqnarray}
	where 	$$\sigma _1^{C}\! =\!\!\! \sqrt {\sum\limits_{i \ne 1} {\sigma _{i,B}^2{K}} \! +\! \sum\limits_{\forall i} {\sigma _{i,D}^2{D}} }, \quad	\mu _1^{C} \!=\! \sum\limits_{i \ne 1} {{\mu _{i,B}}{K}} \! + \!\sum\limits_{\forall i} {{\mu _{i,D}}{D}}.$$
	Similarly, for the $(1,k)$-th D2D receiver the goodput data rate can be presented as
	\begin{equation}
	{\bar {\emph{g}}_{1,k}^{D2D}} = (1-\varepsilon) \times {L} \times {\log _2}\!\!\left(\! {1 + \frac{{{F^{ - 1}}{{( \varepsilon  )}_{{\lambda  = \rho _{1,k}^{1,k}}}}}}{{\mu _1^{D2D}}}}\! \right), \label{goodputD2D}
	\end{equation}
	where $\mu _1^{D2D} = \sum\limits_{i \ne 1} {\left( {{\mu _{i,B}}{K} + {\mu _{i,D}}{D}} \right)}  + {\mu _{1,D}}({D} - 1)$.
	Then the overall average goodput of the $1$-st cell is given as
	\begin{eqnarray}
{{\bar {\emph{g}}}_1} = K{\mathbb{E}_{\bf{U}}}\left[ {{{\bar {\emph{g}}}_{1,k}}} \right] + D{\mathbb{E}_{\bf{U}}}\left[ {\bar {\emph{g}}_{1,k}^{D2D}} \right] \label{Overall},
	\end{eqnarray}
	where ${\mathbb{E}_{\bf{U}}}\left[ {{{\bar {\emph{g}}}_{1,k}}} \right]$ and ${\mathbb{E}_{\bf{U}}}\left[ {\bar {\emph{g}}_{1,k}^{D2D}} \right]$ are expectations based on the distribution of locations for single downlink user's average goodput and single D2D receiver's average goodput.
\end{Theorem}
\begin{proof}
	Theorem 1 is directly derived from lemma 2.
\end{proof}
The interference brought by D2D communication degrades the general performance of both downlink and D2D receivers, as revealed in (\ref{CellularSIR}) and (\ref{D2DSIR}). However, both downlink and D2D receivers with different locations may have different performance loss. Moreover, as the number of D2D links increases, the average goodput of both downlink users and D2D receivers degrades. However, the D2D number increase can improve the cell overall average goodput as revealed in (\ref{Overall}). Thus the cell overall average goodput is determined by the trade-off between the D2D number increase and the downlink users' and D2D receivers' average goodput degradation. As it may be quite complicated to figure out the trade-off by numerical simulation, we show the analytical results in the next section.

\section{Simulation Results}
In the simulation, we consider a network with 19 hexagonal cells each with radius $R=300m$. The number of downlink users in each cell is 10, and the antennas number of BS is 250. The length of pilot sequence is ${L_p} = 31$ and the length of downlink symbols within one frame is ${L}= 50$. The pathloss exponents are $\sigma=3.76$ and $\kappa=4.37$. The transmitting power of each BS and each D2D transmitter is 46 dBm and 23 dBm.

 In Fig. \ref{distance}, there are 10 D2D transmitters in each cell and the transmission distance of each D2D link is $10m$. For both downlink users and D2D receivers, we give the numerical results as well as the asymptotic results. It can be observed that the average goodput of downlink users decreases with respect to its distance to the service BS. This is because of stronger inter-cell interference. On the other hand, the average goodput of D2D links increases with respect to its distance to BS. This demonstrates the impact of interference from BS. Even there are a large number of antennas at the BS, its interference to D2D receiver is still strong and dominant.

 \begin{figure}
 	\centering
 	\includegraphics[height=4.2cm, width=6cm]{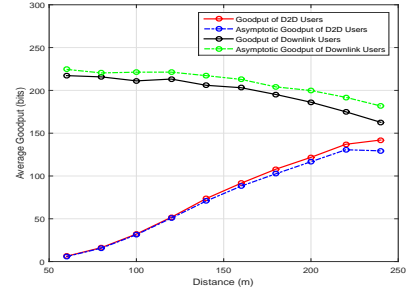}
 	\caption[frame]{Goodputs for downlink users and D2D receivers from different distances to the service BS, where $\varepsilon =0.1 $, D2D number = 10. } \label{distance}
 \end{figure}

 \begin{figure}
 	\centering
 	\includegraphics[height=4.2cm, width=6cm]{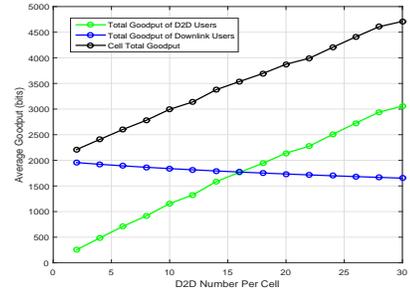}
 	\caption[frame]{Total Goodputs versus different D2D numbers per cell, $\varepsilon =0.1 $. } \label{D2Dnumber}
 \end{figure}

 Fig. \ref{distance} shows that the analytical results fit the actual performance quite well. Thus we can study the effect of D2D number with analytical expressions instead of complicated simulations. In Fig. \ref{D2Dnumber}, the effect of D2D number is demonstrated through analytical results, which show that the cell overall average goodput increases almost linearly with the D2D number, while the average goodput of all downlink users declines. Thus we reveal the trade-off between cell overall average goodput and downlink users' average goodput. Therefore, to achieve decent cell overall average goodput and fair average goodput for downlink users simultaneously, the number of D2D links per cell can be determined through our analytical results.
\section{Conclusion}

In this paper, we focus on the average goodput performance analysis of a D2D underlay downlink massive MIMO system. It is assumed that the cell coverage is hexagonal and the distribution of downlink users and D2D transmitters is independent and uniform. The asymptotic SIR expressions for both downlink transmission and D2D links are firstly derived. To take the potential packet outage, we continue to derive the approximated expressions of average goodput for both downlink and D2D links, which measures the number of information bits successfully delivered to the receiver. Based on it, the trade-off between the two types of links can be studied. Through simulation we show that the analytical results fit the numerical results quite well.

 \begin{appendices}
 	\section{Proof of Lemma 1}
 	In massive MIMO systems, when $M$ is sufficiently large, we usually utilize asymptotic orthogonality of channel as $$\frac{{{\bf{h}}_{j,k}^i{{\left( {{\bf{h}}_{j,k}^i} \right)}^H}}}{M} \! \to \! \rho _{j,k}^i, \! \quad \frac{{{\bf{h}}_{j,k}^i{{\left( {{\bf{h}}_{m,n}^l} \right)}^H}}}{M} \! \to \! 0,(i,j,k)\! \ne \!(l,m,n)$$	
 	It is also the same for ${\bf{v}}_{j,k}^i$. Thus the interference term caused by neighboring BSs in Equation (\ref{MatrixCellular}) can be simplified as
 	\begin{eqnarray}
 	&&\!\!\!\!\!\!\sum\limits_{\forall l \ne i} {{\bf{H}}_i^l} {{\bf{P}}_l}{\bf{W}}_l^d \approx \sum\limits_{\forall l \ne i} { {\bf{H}}_i^l{{({\bf{\hat H}}_l^l)}^H}{{\bf{R}}_l}^{ - 1}{\bf{W}}_l^d}\nonumber \\
 	&\approx& \!\!\!\!\!\! \sum\limits_{l \ne i}\! {{\bf{H}}_i^l{{({\bf{H}}_i^l)}\!^H}\!(\frac{{{\bf{X}}_i^p{{( {{\bf{X}}_l^p} )}\!^H}}}{{{P_u}{L_p}}}\!){{\bf{R}}_l}\!^{ -\! 1}\!{\bf{W}}_l^d} \!+\! \!\sum\limits_{l \ne i}\! {\bf{H}}_i^l{({\bf{H}}_l^l)\!^H}{{\bf{R}}_l}\!^{ -\! 1}\!{\bf{W}}_i^d.\nonumber
 	\end{eqnarray}
 	The first term remains because ${\bf{H}}_i^l{({\bf{H}}_i^l)^H} \gg {\bf{H}}_i^l{({\bf{H}}_m^l)^H},m \ne i$. The second term remains because ${\bf{I}} \gg ({\bf{X}}_l^p{\left( {{\bf{X}}_i^p} \right)\!^H}/{P_u}{L_p})$.
 	Similarly, for D2D receivers in the $i$-th cell, the interference from neighboring BSs can be simplified as
 	\begin{eqnarray}
 	&&\sum\limits_{\forall l \ne i} {{\bf{V}}_i^l} {{\bf{P}}_l}{\bf{W}} \approx \sum\limits_{l \ne i} {\bf{V}}_i^l{({\bf{H}}_l^l)^H}{{\bf{R}}_l}^{ - 1}{\bf{W}}_i^d.\nonumber
 	\end{eqnarray}
 	Therefore, the received signal of the $(1,k)$-th downlink user can be presented as
 	\begin{eqnarray}
 	{\bf{y}}_{1,k} \!\!\!\!\!&=&\!\!\!\!\! \sqrt {{P_b}}{\bf{w}}_{1,k}^d \!\!+\!\!\! \sum\limits_{\forall i \ne 1,j}\!\!\!\! \sqrt {{P_b}}{\rho _{1,k}^i{\bf{x}}_{1,k}^p{{({\bf{x}}_{i,j}^p)}^H}{\rho _{i,j}^i\!\!}^{ - 1}{\bf{w}}_{i,j}^d/({P_u}{L_p})}\nonumber  \\
 	&&\!\! +\!\! \sum\limits_{\forall i \ne 1,j}\sqrt {{P_b}} {\frac{{{{({\bf{h}}_{1,k}^i)}^H}{\bf{h}}_{i,j}^i{\bf{w}}_{i,j}^d}}{M\rho {{_{i,j}^i}}}}  + \sum\limits_{\forall i,j}\sqrt {{P_d}}  {u_{1,k}^{i,j}{\bf{s}}_{i,j}^d}, \label{CellularRec} \nonumber
 	\end{eqnarray}
 	where the intra-cell interference for downlink users is neglected for it is much smaller than inter-cell interference. Similarly, the received signal of the $(1,k)$-th D2D receiver is
 	\begin{eqnarray}
 	{\bf{y}}_{1,k}^{D2D} \approx \sum\limits_{\forall i,j}\sqrt {{P_b}}  {\frac{{{{({\bf{v}}_{1,k}^i)}^H}{\bf{h}}_{i,j}^i{\bf{w}}_{i,j}^d}}{M\rho {{_{i,j}^i}}}}  + \sum\limits_{\forall i,j}\sqrt {{P_d}} {g_{1,k}^{i,j}{\bf{s}}_{i,j}^d}. \label{D2DRec} \nonumber
 	\end{eqnarray}

 	In this paper, we assume the transmit signals are Gaussian with zero mean and unit variance. Then asymptotic SIR expressions as (\ref{CellularSIR}) and (\ref{D2DSIR}) can be straightforward obtained.
 	
 	\section{Proof of lemma2}
 	We first consider the downlink users. The asymptotic SIR expression is presented in (\ref{CellularSIR}). In the dominator, $\left\{\! {\left. {{{\left(\! {\frac{{{\bf{h}}_{1,k}^i{{\bf{h}}_{i,j}^i}^H}}{{M\rho _{i,j}^i}}} \!\right)}^2} \!\!+\! \frac{1}{{{L_p}}}{{\left( {\frac{{\rho _{1,k}^i}}{{\rho _{i,j}^i}}} \right)}^2}} \right|\!\forall i\! \ne\! 1,j}\! \right\}$ are independent variables with expectation ${\mu _{i,B}}$ and variance ${\sigma _{i,B}^2}$. Meanwhile, $\left\{\! {\left. {u_{1,k}^{i,j}{{(u_{1,k}^{i,j})}\!^H}}\!\frac{{{P_d}}}{{{P_b}}} \right|\forall i,j} \!\right\}$ are also independent variables with expectation ${{\mu _{i,D}}}$ and variance ${\sigma _{i,D}^2}$.
 	When the number of interfering cellular users $(C-1)K$ and D2D transmitters $CD$ is sufficiently large, we have $\frac{{\sigma _{i,B}^2}}{{\sum\limits_{\forall i \ne 1,j} {\sigma _{i,B}^2}  + \sum\limits_{\forall i,j} {\sigma _{i,D}^2} }} \to 0,\forall i \ne 1,j$ and $\frac{{\sigma _{i,D}^2}}{{\sum\limits_{\forall i \ne 1,j} {\sigma _{i,B}^2}  + \sum\limits_{\forall i,j} {\sigma _{i,D}^2} }} \to 0,\forall i,j$, which satisfies requirements of the Lindeberg-Feller Central Limit Theorem \cite{Chow1978}. According to the theorem, $1/{\gamma _{1,k}}$ converges to a Gaussian variable when $(C-1)K+CD$ is sufficiently large. Thus, we apply Gaussian approximation on $1/{\gamma _{1,k}}$. Then the probability that the downlink SIR of the $(1,k)$-th cellular user is less than ${{\tau}}$ (coverage outage probability) can be written as
 	\begin{eqnarray}
 	\Pr \left[ {{\gamma _{1,k}} \le {\tau }} \right] 	\approx Q\left[ {\frac{{\frac{{{1}}}{{{\tau}}} - {\mu _{i,B}}{K} - {\mu _{i,D}}{D}}}{{\sqrt {\sigma _{i,B}^2{K} + \sigma _{i,D}^2{D}} }}} \right].\nonumber
 	\end{eqnarray}
 	
 	The probability that the SIR of the $(1,k)$-th D2D link is less than ${{\tau ^{D2D}}}$ can be presented as
 	\begin{eqnarray}
 	\Pr \! \left[ {{\gamma _{1,k}}\! \le\! {\tau ^{D2D}}} \right] \!=\! \Pr\! \left[ {g_{1,k}^{1,k}g{{_{1,k}^{1,k}}^H}\!\! - \!{\tau ^{D2D}}\!\left( {I_{1,k}^{BS} + I_{1,k}^{D2D}} \right) \!\le\! 0} \right]\!,\nonumber
 	\end{eqnarray}
 	where $I_{1,k}^{BS} \!=\!\!\! \sum\limits_{\forall i \ne 1,j}\! {{{\left(\! {\frac{{{\bf{v}}_{1,k}^i{\bf{h}}{{_{i,j}^i}^H}}}{{M\rho _{i,j}^i}}} \!\right)}^2}}\!\frac{{{P_b}}}{{{P_d}}}$, $I_{1,k}^{D2D} \!=\!\!\!\! \sum\limits_{\forall (i,j) \atop \ne (1,k)} { g_{1,k}^{i,j} g{{_{1,k}^{i,j}}^H}}$. Due to D2D receivers are  quite close to their associated transmitters, the variation of the interference ${I_{1,k}^{BS} + I_{1,k}^{D2D}}$ is relatively much smaller than the signal variation. Thus we can regard the interference as a constant. Since the downlink channel ${g_{1,k}^{1,k}}$ is a complex Gaussian variable, the signal power is an exponential variable, where the coverage outage probability can be simplified as
 	\begin{eqnarray}
 	\Pr \!\left[\! {{\gamma _{1,k}} \!\le\! {\tau ^{D2D}}} \right]\! =\! \Pr\! \left[ {g_{1,k}^{1,k}g{{_{1,k}^{1,k}}^H} \!-\! {\tau ^{D2D}}\mathbb{E}\!\left[ {I_{1,k}^{BS} \!+\! I_{1,k}^{D2D}} \right]\! \le\! 0} \right]\nonumber
 	\end{eqnarray}
    Thus (\ref{AsySIRd2d}) can be easily obtained. This finishes the proof.

 \end{appendices}



\bibliographystyle{ieeetr}%
\bibliography{D2d}

\end{document}